\documentclass[11pt]{article}
\usepackage{times}
\usepackage{amsmath,amsthm,amsfonts}

\begin{document}

\title{Short seed extractors against quantum storage}

\author{Amnon Ta-Shma\thanks{Department
of Computer Science, Tel-Aviv University, Tel-Aviv 69978, Israel.
Supported by the European Commission under the Integrated Project
QAP funded by the IST directorate as Contract Number 015848, by
Israel Science Foundation grant 217/05 and by USA Israel BSF grant
2004390. Email: amnon@tau.ac.il. }}

\maketitle

\newcommand{\remove}[1]{}

\newcommand{\set}[1]{{\left\{ #1\right\}}}
\newcommand{\B}[0]{{\left\{0,1\right\}}}
\newcommand{\zo}{\set{0,1}}

\newcommand{\abs}[1]{\left| #1 \right|}
\newcommand{\norm}[1]{\left\| #1 \right\|}
\newcommand{\half}{\frac 1 2}
\newcommand{\minent}{{H_{\infty}}}
\newcommand{\Tr}{\mathrm{Trace}}
\newcommand{\eps}{{\epsilon}}
\newcommand{\logeps}{\log \epsilon^{-1}}
\newcommand{\eqdef}{\stackrel{\rm def}{=}}
\newcommand{\wh}[1]{\widehat{#1}}
\newcommand{\poly}{{\rm poly}}

\newcommand{\ket}[1]{\left|#1\right\rangle}
\newcommand{\bra}[1]{\left\langle #1\right|}
\newcommand{\braket}[2]{\left.\left\langle #1\right|#2\right\rangle}
\newcommand{\ketbra}[2]{\ket{#1}\!\bra{#2}}
\newcommand{\tensor}{{\otimes}}

\newcommand{\enc}[3]{{#1} \stackrel{{#3}}{\mapsto} {#2}}
\newcommand{\bn}{{\bar{n}}}

\newcommand{\cC}{\mathcal{C}}
\newcommand{\cD}{\mathcal{D}}
\newcommand{\cE}{\mathcal{E}}
\newcommand{\rE}{\mathrm{E}}
\newcommand{\cF}{\mathcal{F}}
\newcommand{\cH}{\mathcal{H}}
\newcommand{\cQ}{\mathcal{Q}}
\newcommand{\E}{\mathbb{E}}

\newcommand{\NW}{\mathrm{NW}}
\newcommand{\TR}{\mathrm{TR}}


\newtheorem{theorem}{Theorem}[section]
\newtheorem{definition}{Definition}[section]
\newtheorem{lemma}{Lemma}[section]
\newtheorem{claim}{Claim}[section]
\newtheorem{proposition}{Proposition}[section]
\newtheorem{corollary}{Corollary}[section]
\newtheorem{example}{Example}[section]
\newtheorem{fact}{Fact}[section]

\begin{abstract}
Some, but not all, extractors resist adversaries with limited quantum storage. In this paper we show that Trevisan's extractor has this property, thereby showing an extractor against quantum storage with logarithmic seed length.
\end{abstract}

\section{Introduction}
\label{sec:introduction}
In the classical \emph{privacy amplification} problem Alice and Bob share information that is only partially secret towards an eavesdropper Charlie. Their goal is to distill this information to a shorter string that is completely secret. The problem was introduced in \cite{BBR88,BBCM95}. The classical privacy amplification problem can be solved almost optimally using extractors.\footnote{Extractors are defined in Section \ref{sec:extractors-against-quantum-storage}.}

An interesting variant of the problem, where the eavesdropper Charlie is allowed to keep quantum information, was introduced by Konig, Maurer and Renner \cite{KMR03,KMR05}. Let us call such an extractor \emph{an extractor against quantum storage}.\footnote{A formal definition is given in Section \ref{sec:extractors-against-quantum-storage}.} This situation naturally occurs in analyzing the security of some quantum key distribution (QKD) protocols and in bounded-storage cryptography. For example, \cite{CRE04} show a generic way of using extractors against quantum storage to prove the security of certain QKD protocols. Using extractors for bounded-storage cryptography demands more from the extractor (it should be "locally computable"), but also allows more specific assumptions about the source distribution (e.g., \cite{KT08} and \cite{KR07}).

Special cases of the problem are also of great interest. The first such example appears in \cite{ANTV99,N99,ANTV02} where random access codes are studied. Alice and Bob share a random length $n$ string $x$ on which the eavesdropper Charlie knows $b$ bits of information. If Charlie is classical, then choosing a random $i \in [n]$ and outputting $x_i$ results in an almost uniform bit. The question studied in the above papers is wether the same also holds when Charlie is quantum and may hold $b$ \emph{quantum bits}. It was shown in \cite{ANTV99,N99,ANTV02} that the answer is positive, and this gives an extractor against quantum storage, albeit, with a \emph{single} output bit.

Konig, Maurer and Renner \cite{KMR03,KMR05} show that the pair-wise independent extractor of \cite{ILL89} is also good (and with the same parameters) against quantum storage. Using the same techniques the result can also be extended to using almost pair-wise independence \cite{SZ99,GW97}. Another classical extractor for very high min-entropies was shown to hold against quantum storage in \cite{FS07} (the classical version appears, e.g., in \cite{DS05}). Konig and Terhal \cite{KT08} showed that any single output extractor is also good against quantum storage. They also showed that any extractor with error $\eps$, has at most $2^{O(b)}\eps$ error against $b$ quantum storage. Thus, if some extractor has a good dependence on the error (as is often the case) one can make the extractor good against $b$ quantum storage by taking a longer seed (often, longer by only $O(b)$ bits).

It is tempting to conjecture that every extractor against classical storage should also be good against quantum storage. However, Gavinsky et. al. \cite{GKKRW07} show an example of an extractor that works well against classical storage but fails even against much shorter quantum storage.

To summarize, many techniques and constructions generalize and work well against quantum storage. Yet, in spite of much effort, none of the above methods give a short seed extractor against quantum storage. \cite{KMR03,KMR05} have seed length $\Omega(n)$ and the variant with almost pair-wise independence has seed length $\Omega(m)$, where $n$ is the length of $x$ and $m$ is the output length. \cite{FS07} requires the seed length to be $\Omega(b)$ where $b$ is the bound on the quantum storage. \cite{KT08} show any single output bit extractor is good against quantum storage, and for $m$ bits their method gives $m \log n$ seed length. Alternatively, they show one can do with $O(\log n+b)$ seed length, which is again not applicable if $b$ is relatively large (say, super-polynomial). In contrast, classically, there are many explicit constructions with poly-logarithmic seed length, some even with logarithmic seed length. Some of these constructions are summarized in Table \ref{table:classical-extractors}. A natural question that repeatedly appears in the above mentioned papers is whether one can show a logarithmic seed length extractor against quantum storage.

In this work we show that Trevisan's extractor \cite{T01} is also good against quantum storage, with somewhat weaker parameters.

\begin{theorem}
\label{thm:main}
There exists a constant $c>1$, such that for every $k,b<n$ and $\eps>0$ there exists an explicit $(k,b,\eps)$ strong extractor $E:\B^n \times \B^t \to \B^m$ against $b$ quantum storage, with seed length $t=O({\log^2 n \over \log m})$ and output length
$m=\Omega({\eps \over \log n} ({k \over b})^{1/c})$.
%
%
%
%
%
\footnote{The constant $c$ we currently achieve is $c=15$.}
\end{theorem}

Plugging $k=n$ which is the usual setting for privacy amplification, we get:

\begin{corollary}
\label{cor:main}
For the above constant $c$, for every $\beta < 1, \gamma < {1-\beta \over c}$ there exists an explicit $(n,b=n^\beta,\eps=n^{-\gamma})$ strong extractor $E:\B^n \times \B^t \to \B^m$ against quantum storage, with output length $n^{\Omega(1)}$ and seed length $t=O(\log n)$.
\end{corollary}

The seed length is $O(\log n)$ and matches classical extractor's lower bound up to constant multiplicative factors. The error $\eps$ is not that good, as it can not get below, e.g., $1/k$. The number of extracted bits is $n^{\Omega(1)}$. This should be compared with $n^{1-\zeta}$ for $\zeta$ arbitrarily small, in Trevisan's extractor against classical storage. Thus we have a polynomial loss here compared to the original classical scheme.

Table \ref{table:classical-extractors} summarizes the parameters of the known classical extractors against quantum storage. Our work gives the first solution to the privacy amplification problem against quantum storage with logarithmic seed length. We believe that other extractor constructions should also be good against quantum storage.

{\bf The technique.} One way to view Trevisan's extractor is as follows. We already said a random access code is a classical extractor outputting a \emph{single} bit. One can take $m$ \emph{independent} copies of this extractor and get an extractor outputting $m$ bits. The price of this is that the seed length becomes $\Omega(m)$. To fix this, in Trevisan's extractor a short seed of length $O(\log n)$ is used to create $m$ sets that are \emph{pair-wise nearly-disjoint}. The analysis shows that in the classical setting the $m$ nearly-disjoint sets can replace the $m$ independent sets, resulting with $m$ output bits but only $O(\log n)$ seed length.

Can this also work against quantum storage?
Anbainis et al. \cite{ANTV02} show a random access code is a single-output extractor against quantum storage. Konig and Terhal \cite{KT08} show taking $m$ \emph{independent} copies of this extractor is good against quantum storage. What about the derandomized version with pair-wise nearly-disjoint sets? Is it also good against quantum storage?

The analysis of Trevisan's extractor uses the fact that it is built upon a \emph{reconstructible} pseudo-random generator (PRG). Loosely speaking, in such structures any mechanism that breaks the extractor (i.e., distinguishes its output $E(x,U)$ from uniform) can be used together with a short advice to reconstruct its input $x$. This kind of reasoning looks well suited to generalizations to extractors against quantum storage. Assume Charlie can distinguish the extractor output $E(x,U)$ from uniform using $b$ qubits of storage. Then, the reconstruction property tells us we should be able to reconstruct $x$ using Charlie's reconstruction procedure, his $b$ qubits of information and a short advice of $a$ classical bits. Thus, we can reconstruct $x \in \B^n$ using only $a+b$ qubits. Basic Quantum information theory tells us then that $a+b \ge n$, or putting it differently, whenever $b < n-a$, we output uniform bits.

A fundamental problem that arises in the proof is that quantum advice is fragile, and using it once degrades it. This is exactly the main problem dealt with in \cite{ANTV99,N99,ANTV02}. Simplifying things, this problem forces the reconstruction algorithm to making only few queries to Charlie. Thus, a key ingredient in our solution is replacing the error correcting codes used in Trevisan's extractor with locally list-decodable codes (see Section \ref{sec:soft-decision-local-list-decoding}). Another problem is that the analysis requires random access codes of \emph{subsets}. We explain the technical  problems we encounter and their solution (and the way this affects the parameters) in detail in the technical sections.

\hspace*{-2cm}
\begin{table}
\begin{tabular}{|l|l|l|l|}\hline
no. of  truly&
no. of &
Against classical storage & Against quantum storage\\
random bits &
output bits & & \\
\hline\hline
 $O(n)$ &
$m=n-b-O(1)$ & Pair-wise independence, \cite{ILL89} & \checkmark \cite{KMR03} \\
$O(b+\log n)$ &
$m=n-b-O(1)$ & Fourier analysis, collision \cite{DS05} & \checkmark \cite{FS07} \\
 $\Theta(m)$ &
$m \le n-b-O(1)$ & Almost pair-wise ind., \cite{SZ99,GW97} & \checkmark, based on  \cite{KMR03}\\
 $O({\log^2n \over \log (n-b)})$ & $(n-b)^{1-\alpha}$ &
Designs, \cite{T01} & \checkmark, This paper. $m \approx {\eps \over \log n}({n-b \over b})^{\Omega(1)}$\\
$O(\log n)$ & $m=\Omega(n-b)$ & \cite{LRVW03,GUV07,DW08} & ?
\\
\hline\hline
\end{tabular}
\caption{Milestones in building explicit strong extractors against $b$ storage, in the classical and quantum setting. The error $\epsilon$ is a constant.}
\label{table:classical-extractors}
\end{table}

\section{Preliminaries}
\label{sec:preliminaries}

We begin with some standard notation. A distribution $D$ on $\Lambda$ is a function $D:\Lambda \to [0,1]$ such that $\sum_{a \in \Lambda} D(a)=1$. $x \in D$ denotes sampling according to the distribution $D$. $U_t$ denotes the uniform distribution over $\B^t$. We measure distance between two distributions with the variational distance $d(D_1,D_2)=\half |D_1-D_2|_1 = \half \sum_{a \in \Lambda} |D_1(a)-D_2(a)| = \max_{S \subseteq \Lambda} D_1(S)-D_2(S)$, where $D(S)=\sum_{s \in S} D(s)=\Pr_{a \in D} (a \in S)$.

The entropy of $D$ is $H(D)=\rE_{a \in D} \log(1/D(a))$. The min-entropy of $D$ is $\minent(D)= \min_{a: D(a)>0} 1/\log(D(a))$. If $\minent(D) \le k$, then for all $a$ in its support $D(a) \ge 2^{-k}$. A distribution is flat if it is uniformly distributed over its support. For flat distributions $\minent(X)=H(X)$. Every distribution $X$ with $\minent(X) \ge k$ can be expressed as a convex combination $\sum \alpha_i X_i$ of flat distributions $X_i$ each with min-entropy at least $k$.

A superposition is a vector in some Hilbert space. $\cH_{2^b}$ denotes a Hilbert space of dimension $2^b$.
A general quantum system is in a {\em mixed state\/}---a
probability distribution over superpositions. Let $\{p_i, \ket{\phi_i}\}$ denote the mixed state where
superposition~$\ket{\phi_i}$ occurs with probability~$p_i$. The
behavior of the mixed state $\set{p_i,\ket{\phi_i}}$ is completely characterized by its {\em
density matrix\/}~$\rho = \sum_i p_i \ketbra{\phi_i}{\phi_i}$ in the sense that
two mixed states with the same density matrix have the same
behavior under any physical operation. Notice that a density matrix over a Hilbert space $\cH$ belongs to $Hom(\cH,\cH)$, the set of linear transformation from $\cH$ to $\cH$. Density matrices are positive semi-definite operators and have trace $1$.

A POVM (Positive Operator Value Measure) is the most general formulation of a measurement in quantum computation. A POVM on a Hilbert space $\cH$ is a collection $\set{E_i}$ of positive semi-definite operators $E_i:Hom(\cH,\cH) \to Hom(\cH,\cH)$ that sum-up to the identity transformation, i.e., $E_i \succeq 0$ and $\sum E_i=I$. Applying a POVM $\set{E_i}$ on a density matrix $\rho$ results in answer $i$ with probability $\Tr(E_i \rho)$. 

\section{Extractors against quantum storage}
\label{sec:extractors-against-quantum-storage}

\subsection{Extractors and privacy amplification}

Alice holds a string $x$ drawn from the uniform distribution. An adversary $C$
is given some partial information about $x$ in two ways:

\begin{itemize}
\item
First, $C$ is told a small subset $X \subseteq \B^n$ from which the input $x$ is taken.
\item
Second, we let $C$ keep $b$ bits of information about $x$.
\end{itemize}

In the classical world we model the second item by two arbitrarily correlated random variables $X$ and $C$, with the constraint that $C$ is distributed over $\B^b$. In the quantum world, we say an $(n,b)$ quantum encoding is a collection $\set{\rho(x)}_{x \in \B^n}$ of density matrices $\rho(x) \in \cH_{2^b}$, and we let $C$ hold any $(n,b)$ quantum encoding of $X$.

Our goal is to find a function $E:\B^n \times \B^t \to \B^m$ such that $E(X,U_t)$, which is the distribution obtained by picking $x \in X, y \in U_t$ and outputting $E(x,y)$, "looks uniform" to the adversary $C$. We define this as follows. We say a boolean test $T$ $\epsilon$--distinguishes $D_1$ from $D_2$ if $|\Pr_{x_1 \in D_1}[T(x_1)=1] - \Pr_{x_2 \in D_2}[T(x_2)=1]| \ge \epsilon$. We say $D_1$ is $\eps$-indistinguishable from $D_2$ if no boolean POVM can $\eps$ distinguish $D_1$ from $D_2$. We define:

\begin{definition}
A function $E:\B^n \times \B^t \to \B^m$ is a $(k,b,\eps)$ strong extractor against quantum storage,
if for any distribution $X \subseteq \B^n$ with $\minent(X) \ge k$ and every $(n,b)$ quantum encoding $\set{\rho(x)}$,
$U_t \circ E(X,U_t) \circ \rho(X)$ is $\eps$-indistinguishable from $U_{t+m} \circ \rho(X)$.\footnote{$U_t \circ E(X,U_t) \circ \rho(X)$ denotes the mixed state obtained by sampling $x \in X, y \in \B^t$ and outputting $\ket{y,E(x,y)} \tensor \rho(x)$. Similarly, $U_{t+m} \times \rho(X)$ denotes the mixed state obtained by sampling $w \in \B^{t+m}, x \in X$ and outputting $\ket{w} \tensor \rho(x)$.}
\end{definition}

In the definition we could have replaced the condition "for any distribution $X \subseteq \B^n$ with $\minent(X) \ge k$" with the condition "for any \emph{flat} distribution $X \subseteq \B^n$ with $\minent(X) \ge k$", as any distribution  $X \subseteq \B^n$ with $\minent(X) \ge k$ can be expressed as a convex combination of flat distributions with min-entropy $k$.

We similarly define a $(k,b,\eps)$ strong extractor against \emph{classical} storage, where we allow the adversary $C$ two types of information: first we tell $C$ that $x$ is drawn from a small subset $X \subseteq \B^n$, and second, we let $C$ store $b$ bits of information about $x$. However, classically, these two types of information are redundant. Formally,

\begin{lemma}
Let $E:\B^n \times \B^t \to \B^m$. Let $k \ge b \ge 0$ and $\eps \ge 0$.
If $E$ is a $(k-b-\logeps,\eps)$ strong extractor then $E$ is a $(k,b,2\eps)$ strong extractor against classical storage.
\end{lemma}

\begin{proof}
Let $X$ be a flat distribution over $2^k$ elements.
Assume $C$ keeps $b$ bits of information. Except for probability $\eps$, $C$ gets a value $c$ such that $\Pr[C=c]\ge \eps 2^{-b}$ and so $\minent(X|C=c) \ge k-b-\logeps$ and therefore $(U_t \circ E(X,U_t)~|~C=c)$ is $\eps$ close to uniform. Thus $E$ is a $(k,b,2\eps)$ strong extractor against classical storage.
\end{proof}

A $(k-b-\logeps,\eps)$ extractor is not necessarily a $(k,b,2\eps)$ strong extractor against \emph{quantum} storage. One formal reason is that it is not clear how to define the conditional distribution $(X|C=\rho)$ when $C$ may be quantum. Renner \cite{R05} defines \emph{smooth min-entropy} for this case, but still it is not clear how to define the marginal distribution itself as it depends on which measurement $C$ chooses to take later.

Another way to look at the problem is as follows. In the classical world, $C$ has to first choose $c$ bits of information about $x$ which already determines a distribution $(X|C=c)$, and only then an independent random seed $y \in \B^t$ is chosen and $E(x,y)$ is calculated. In the quantum world, however, things are not that simple. $C$ first chooses $c$ qubits of information about $x$. This by itself does not determine any classical distribution $X$ on $\B^n$. Next, an independent random seed $y \in \B^t$ is chosen and $E(x,y)$ is calculated. Finally, $C$ may choose which measurement to make based on $x$ and $y$.
The problem is that it may be possible for $C$ to make a measurement that will correlate the distribution $X$ with the seed $y$, making the extractor useless. This point of view is further explained in \cite{KT08}.

\subsection{Random access codes}

A similar problem to the one above appears in \emph{random access codes}. We now explain what random access codes are, as this will turn out to be a basic building block in our result. A fundamental result in quantum information theory, Holevo's theorem~\cite{H73}, states that no more than $b$ classical bits of information can be faithfully transmitted by transferring $b$ quantum bits from one party to another. Formally,

\begin{theorem} (Holevo)
Let $\set{\rho(x)}$ be any $(n,b)$ quantum encoding. Let $X$ be a random variable with distribution
$\set{p_x}$ and let $\rho(X) = \rE_x \rho(x) = \sum_x p_x \rho_x$.
If $Y$ is any random variable obtained by performing a measurement on
the encoding, then $I(X:Y) \le  S(\rho(X)) - \rE_x S(\rho_x) \le S(\rho(X))$.
\end{theorem}

In view of this result, it is tempting to conclude that the exponentially many degrees of freedom latent in the description of a quantum system must necessarily stay hidden or inaccessible. However, the situation is more subtle since the recipient of the~$n$ qubit quantum state has a choice of measurement he can make to extract
information about their state. In general, these measurements do not commute. Thus making a particular measurement will disturb the system, thereby destroying some or all the information that would have been revealed by another possible measurement. Indeed, Ambainis et. al. \cite{ANTV99} ask whether there exists an $(n,b)$ quantum encoding $\set{\rho(x)}$ such that the recipient can learn any bit $x_i$ of his choice. I.e., they define:

\begin{definition}\cite{ANTV02}
A $\enc{n}{t}{p}$ quantum {\em random access\/} encoding is an $(n,t)$ encoding $\set{\rho(x)}_{x \in \B^n}$ such that for every~$1 \le i \le n$,
there is a POVM $\cE^i=\set{\cE^i_0,\cE^i_1}$ (i.e., $\cE^i_0+\cE^i_1=I, \cE^i_j \succeq 0$)
such that for all $x \in \{0,1\}^n$ we have
$\Tr(\cE^i_{x_i} f(x)) \ge p$.
\end{definition}

\cite{N99,ANTV02} show that any quantum $\enc{n}{t}{p}$ encoding must have $t \ge (1-H(p))n$. In fact, this lower bound also holds if we relax the worst-case condition
$\forall_x \forall_i \Tr(\cE^i_{x_i} f(x)) \ge p$ and replace it with the average-case condition $\forall_x \E_i \Tr(\cE^i_{x_i} f(x)) \ge p$.

In this paper we need random access codes that are defined for \emph{subsets} of $\B^n$. Namely,

\begin{definition}
Let $\cF \subseteq \B^n$. A $\enc{\cF}{t}{p}$ quantum {\em random access\/} encoding is an $(n,t)$ encoding $\set{\rho(x)}_{x \in \cF}$ such that for every~$1 \le i \le n$,
there is a POVM $\cE^i=\set{\cE^i_0,\cE^i_1}$ (i.e., $\cE^i_0+\cE^i_1=I, \cE^i_j \succeq 0$)
such that for all $x \in \cF, i \in [n]$ we have
$\Tr(\cE^i_{x_i} f(x)) \ge p$.
\end{definition}

We prove:

\begin{theorem}
\label{thm:random-access-codes}
Let $\delta \ge 0$, $\cF \subseteq \B^n$.

\begin{enumerate}
\item
Any quantum $\enc{\cF}{t}{\half+\delta}$ encoding satisfies~$t \ge \Omega({\delta^2 \over \log n} \cdot \log |\cF|)$.
\item
\label{thm:random-access-codes:item-2}
Any quantum $\enc{\cF}{t}{1-\delta}$ encoding satisfies~$t \ge \Omega({\log(1/4 \delta) \over \log n} \cdot \log |\cF|)$.
\end{enumerate}
\end{theorem}

\begin{proof}
We use the proof technique of \cite{ANTV99}. First, one can turn the $\enc{\cF}{t}{\half+\delta}$ encoding
into another $\enc{\cF}{O(t \times T)}{1-\eps}$ encoding, with $T=O(\logeps/\delta^2)$,
as follows. The new encoding is $T$ copies of the original encoding. The decoding is the majority vote over the $T$ decodings of the $T$ copies. By Chernoff, The probability of error is at most $\eps$.

Fix $\epsilon={c \over n^2}$ for some constant $c$ that will be fixed later.
Consider some $f \in \cF$ and its encoding $\rho=\rho(f)$. For every $i \in [n]$
the measurement $\cE^i$ recovers $f_i$ with probability at least $1-\eps$, i.e., almost with certainty. It is shown in \cite{ANTV99},\footnote{Implicit in the proof of Lemma 4.2.} that applying sequentially the measurements $\cE^1,\ldots,\cE^n$
results in a distribution $Y$ that outputs $(f_1,\ldots,f_n)$ with probability at least $1-4n\sqrt{\eps}=1-4\sqrt{c}$. Taking $c$ small enough, we recover $y$ with probability $\half$. By Holevo's theorem, $Tt \ge I(U_\cF:Y) \ge \half \log(|\cF|)$.

For the second item notice that one can turn a $\enc{\cF}{t}{1-\delta}$ encoding
into another $\enc{\cF}{O(t \times T)}{1-\eps}$ encoding, using $T=2\log_{4\delta}\eps$, and the rest is as before.
\end{proof}

Oded Regev showed us an example where the bound in Theorem \ref{thm:random-access-codes} is tight. Partition the $n$ bits to $\sqrt{n}$ blocks each of size $\sqrt{n}$. Take the set $\cF$ to be all bit strings containing exactly one 1 in each block. $\cF$ has
$\Theta(\sqrt{n} \cdot \log n)$ entropy.
Yet, consider the following RAC that uses only $O(\sqrt{n}+\log n)$ bits.
Given $f \in \cF$, with indices $i_1,\ldots,i_{\sqrt{n}}$ (i.e., index $i_j$ is $1$ in the $j$'th block) the RAC encodes $f$ by $(h,h(i_1),\ldots,h(i_k))$, where $h:[\sqrt{n}] \to [10]$ is randomly chosen from a family of pairwise independent
hash functions. When asked for a bit $t$ of the input, say, from the $j$'th block, the decoder just checks whether $h(t)=h(i_j)$. It outputs 1 if yes, otherwise 0. By the pairwise independent property, we output the correct answer with
probability $2/3$ for each question.

We proved Theorem \ref{thm:random-access-codes} with the definition that is worst-case over $i$. We remark that the average case version is false. For example, if $\cF$ is the set of all $n$ bit strings of weight at least ${2 \over 3}n$, there is a trivial random access code of length zero that for all $f \in \cF$ succeeds on average over $i$
with probability at least $2/3$. Thus, here there is a crucial difference between worst-case and average-case complexity over $i$.

\section{Local list-decoding}
\label{sec:soft-decision-local-list-decoding}


A code is a function $\cC:\Sigma^n \to \Sigma^\bn$. We identify a binary code $\cC$ with its image $\cC = \set{ \cC(x) ~|~ x \in \Sigma^n}$. The distance $d$ of the code is the minimum distance between two codewords in $\cC$. The balls of radius ${d-1 \over 2}$ around codewords are disjoint, and therefore one can uniquely correct up to so many errors. If we allow more than $d/2$ errors several decodings are possible. In many cases one can allow almost up to the distance errors and still get only few possible decodings. We say $\cC$ is $(p,L)$ list-decodable if for every $z \in \Sigma^{\bn}$ there are at most $L$ codewords $y$ such that $ag(z,y) \eqdef |\set{i \in [\bn] | z_i=y_i}| \ge p\bn$.

As always one can study the combinatorial properties of a code, or ask for an explicit decoding algorithm. If the decoding algorithm makes only few queries to the corrupted word, we say it is \emph{local}. Formally,

\begin{definition}(local list-decoding)
Let $\cC:\Sigma^n \to \Sigma^\bn$.
We say $\cC$ has a $(p,L,q,\beta)$ local list-decoding if:

\begin{itemize}

\item
$\cC$ is $(p,L)$ list-decodable.

\item
 There exists a probabilistic, polynomial time oracle machine $A$ that on input $k \in [L]$ and $i \in [n]$ outputs a value $A^*(k,i) \in \B$. $A$ can make at most $q$ queries and each query is in the range $[\bn]$.

\item
For every deterministic function $y: \Sigma^\bn \to \Sigma$ and every $x \in \Sigma^n$ such that $ag(y,\cC(x)) \ge p\bn$, there exists $k \in [L]$ such that for every $i \in [n]$, $\Pr_{A} [A^y(k,i)=x(i)] \ge \beta$.
\end{itemize}
\end{definition}

Sudan, Trevisan and Vadhan proved:

\begin{theorem}\label{thm:stv}\cite{STV01}
For every $\delta=\delta(n)>0$, there exists an explicit $[\bn,n]_2$ binary code with output length $\bn=poly(n,{1 \over \delta})$ and $poly(\bn)$ encoding time, that is $(p=\half+\delta,L=poly(\bn),q=poly(\log n,{1 \over \delta}),\beta=1-\delta)$ local list-decodable.\footnote{The code in \cite{STV01} is Reed Muller concatenated with Hadamard. The list-decoding algorithm first list-decodes the Hadamard code, and then uses the result to list-decode the Reed Muller code. As the Hadamard list decoding returns a list, it is better to use there list recovery. Working out the parameters we get field size $|F|$ that is $|F|=O({\log^2 n \over \delta^5})$.
With $|F|^3$ queries the algorithm solves the local list-decoding problem, worst-case over $i$.
We remark that using a better inner code the query complexity can be reduced.}
\end{theorem}

In our case we do not have access to a deterministic function $y: [\bn] \to \Sigma$, but rather to a probabilistic procedure that has high on average success probability.
We are given a probabilistic oracle $O:[\bn] \to \Sigma$. For $y:[\bn] \to \Sigma$ define $ag(O,y) \eqdef \Pr_{i \in [\bn],O} (O(i)=y(i))$. We would like to do local list-decoding when given access to $O$. Formally,

\begin{definition}(probabilistic oracle, local list-decoding)
Let $\cC:\Sigma^n \to \Sigma^\bn$.
We say $\cC$ has a $(p,L,q,\beta)$ \emph{probabilistic oracle}, local list-decoding if:

\begin{itemize}

\item
$\cC$ is $(p,L)$ list-decodable.

\item
 There exists a probabilistic, polynomial time oracle machine $A$ that on input $k \in [L]$ and $i \in [n]$ outputs a value $A^*(k,i) \in \B$. $A$ can make at most $q$ queries and each query is in the range $[\bn]$.

\item
For every probabilistic oracle $O: \Sigma^\bn \to \mathcal{D}$ and every $x \in \Sigma^n$ such that $ag(O,\cC(x)) \ge p\bn$, there exists $k \in [L]$ such that for every $i \in [n]$, $\Pr_{A} [A^O(k,i)=x(i)] \ge \beta$.
\end{itemize}
\end{definition}

If we are just interested in list-decoding (with no restriction on the number of queries) then list decoding a probabilistic oracle is essentially the same as list decoding a string. This is because we can take $O$, and for every query $j \in [\bn]$ sample $y_j=O(j)$. By Chernoff, with high probability, the sampled string $y$ also has high agreement with $\cC(x)$ and therefore the string $x$ appears somewhere in the output list of $y$.

The above argument does \emph{not} work for \emph{local} list-decoding. Here we need the index $k$ to depend on $O$ alone, and not on the sampled string $y$ or the index $i$. This is an essential requirement, as in local list-decoding we do not reconstruct the whole string $x$, but rather a single bit $x_i$ of it. The above argument therefore does not work, as it may happen that the index of $x$ in the list of $y$ depends on the sampled string $y$, and not just on $O$ as required by the definition.

Luckily, going back to the construction of \cite{STV01} one can check that essentially the same analysis shows that:\footnote{This is because the advice for $x$ is a point $v$ and a value $\sigma$ such that $\wh{x}(v)=\sigma$, were $\wh{x}$ is the low-degree extension of $x$, and with high probability such an advice separates for \emph{most} of the sampled strings $y$, the true codeword $\cC(x)$ from the other codewords that arise from $y$.}


\begin{theorem}\label{thm:stv:probabilistic-oracle}(based on \cite{STV01})
For every $\delta=\delta(n)>0$, there exists an explicit $[\bn,n]_2$ binary code with output length $\bn=poly(n,{1 \over \delta})$ and $poly(\bn)$ encoding time, that is $(p=\half+\delta,L=poly(\bn),q=poly(\log n,{1 \over \delta}),\beta=1-\delta)$ probabilistic oracle, local list-decodable.
\end{theorem}

\section{Black-box PRGs}
\label{sec:black-box}

Trevisan showed that good classical black-box PRGs give rise to good
classical extractors. In this section we show that good classical black-box PRGs with \emph{few queries} give rise to good
classical extractors against \emph{quantum storage}.



We begin with a purely classical definition:

\begin{definition}(black-box PRG)
Let $G^{f:[n] \to \B}:\B^t \to \B^m$ be a classical oracle machine with oracle calls to a function $f:[n] \to \B$. $(G^f,R)$ is a black-box $(\epsilon,p)$-PRG with $a$ advice bits and $q$ queries, if:

\begin{itemize}
\item
$R$ is a classical oracle circuit $R(adv,i)$ with inputs $adv \in
\B^a$ and $i \in [n]$. Also, $R$ makes at most $q$ queries to $T$.

\item
For every Boolean function $f: [n] \to \B$, and every
probabilistic oracle $T$ that $\epsilon$--distinguishes $U_t \circ
G^f(U_t)$ from uniform, there exists an advice
$adv=adv(T,f) \in \B^a$ such that for all $i \in [n]$, $\Pr_{R,T} [R^T(adv,i)=f(i)] \ge p$.
\end{itemize}
We call $R$ the \emph{reconstruction algorithm}.
Sometimes we omit $R$ and say $G^f$ is a black-box $(\epsilon,p)$-PRG, meaning that there exists some reconstruction algorithm such that $(G^f,R)$ is a black-box $(\epsilon,p)$-PRG.
\end{definition}

%
%

Trevisan \cite{T01} showed that black-box
pseudorandom generators give rise to extractors. We show they actually give rise to extractors against quantum storage, alas their quality depends on the number of oracle calls in the reconstruction algorithm.

\begin{proposition}\label{prop:black-box:extractor-against-quantum-storage}
(generalizing \cite{T01})
Let $G^f,R$ be as above. Suppose $(G^f,R)$ is a black-box $(\eps,p=1-\delta)$-PRG with $a$ advice bits and $q$ queries. Then $E:\B^n \times \B^t \to \B^m$ defined by
$E(f,y)=G^f(y)$ is a $(k,b,2\epsilon)$ strong extractor against quantum storage, for
$k=\Omega({\log n \over \log(1/4\delta)}(a+qb)) +\logeps$.
\end{proposition}

\begin{proof}
Let $T$ be a quantum test using $b$ qubits of side information $\rho$. We currently think of $T$ as a probabilistic oracle. Let
$\cF$ be the set of all functions $f \in \B^n$ for which $T$
$\eps$-distinguishes $U_t \times E(f,U_t) \times \rho(f)$ from $U_t
\times U_m \times \rho(f)$. We will show $|\cF| = 2^{O((a+qb) \cdot \log(n)/log(1/4\delta))}$. It will then follow that for any $X \subseteq \B^n$,
$|\Pr [T(U_t \times E(X,U_t) \times \rho(X))=1] -\Pr [T(U_t
\times U_m \times \rho(X))=1]|
\le \E_{x \in X}|\Pr [T(U_t \times E(x,U_t) \times \rho(x))=1] -\Pr [T(U_t \times U_m \times \rho(x))=1]| \le \eps+\Pr_{x \in X}[x \in \cF]$. Thus,
$E$ is a $(\log{|\cF| \over \eps},b,2\eps)$ strong extractor against quantum storage.

We now show $\cF$ is indeed small. For any $f \in \cF$, given the right advice $adv=adv(T,f) \in \B^a$ the circuit $R^T(adv,\cdot)$ computes $f:[n] \to \B$ with $q$ queries to $T$ and worst-case (over $i$) success probability $p$. We replace each of the $q$ queries to $T$ with a quantum circuit acting on its classical input and an independent $b$-qubit state that is initialized to $\rho(f)$. Thus, altogether, the new circuit uses $qb$ qubits of side information. Notice that because the inputs to the different queries are in product state, the answers to the $T$ queries are independent.  The resulting quantum circuit recovers the bits of $f:[n] \to \B$ with probability $p$ (\emph{worst-case} over $i$).
Thus, $\cF$ has a random access code of length $a+qb$ and worst-case success $p=1-\delta$. By Theorem
\ref{thm:random-access-codes}, item (\ref{thm:random-access-codes:item-2}), $a+qb = \Omega( {\log(1/4\delta) \over \log n} log |\cF| )$ as desired.
\end{proof}

Thus, we reduced the problem of finding extractors against quantum storage to the classical question of finding good black-box PRG with \emph{few queries}. In the next section we will prove:

\begin{theorem}\label{thm:black-box-PRG:few-queries}
Let $\eps>0$, $m\le n$.
There exists an explicit black-box $(\eps,1-{\eps \over 2m})$ PRG $G^{f:[n] \to \B}:\B^t \to \B^m$ with $a=O(m^2+\log {n \over \eps})$ advice bits, seed length $t=O({\log ^2{n \over \eps} \over \log m})$ and $q=\poly(\log n,{m \over \eps})$ queries.
\end{theorem}

Plugging Thm \ref{thm:black-box-PRG:few-queries} into Proposition \ref{prop:black-box:extractor-against-quantum-storage} we get
Theorem \ref{thm:main}.

\subsection{A black-box PRG with few queries}

Trevisan's PRG \cite{T01} is based on the Nisan-Wigderson PRG \cite{NW94}, which
has a \emph{good on average} reconstruction algorithm. Formally,

\begin{definition}
Let $G^f,R$ be as above.
$(G^f,R)$ is a black-box $(\epsilon,p)$-PRG with
\emph{average-case} reconstruction with $a$ advice bits and $q$ queries,
if for every Boolean function
$f: [n] \to \B$, and every probabilistic oracle $T$ that
$\epsilon$--distinguishes $U_t \circ G^f(U_t)$ from uniform, there
exists an advice $adv=adv(T,f) \in \B^a$ such that $R^T(adv,x)$ makes at most $q$ queries to $T$ and $\Pr [R^T(adv,i)=f(i)] \ge p$, where the probability is over a
uniform $i \in [n]$ and the internal coins of $R$ and $T$.
\end{definition}

The $\NW$ PRG is a black-box PRG with average-case reconstruction. Specifically, for every $\epsilon>0$, $\NW^{f:[n] \to \B}:\B^t \to \B^m$
has $(\epsilon,p=\half+{\epsilon \over 2m})$ average-case
reconstruction with $a=O(m^2)$ advice bits and $t=O({\log^2 n \over \log m})$. The
$\NW$ reconstruction algorithm uses exactly one oracle call to the
distinguishing algorithm. Trevisan used that to prove the following:

\begin{lemma}(Trevisan's worst-case to average-case reduction for black-box PRG)
\label{lem:black-box:average-to-worst-case:Trevisan}
Assume $(G^f,R)$ is a black-box $(\epsilon,\half+\delta)$-PRG with
average-case reconstruction using $a$ advice bits. Further assume the reconstruction algorithm $R$ is deterministic. Let $\cC[\bn,n]_2$ be
a $(\half+\delta,L)$ list-decodable code.
Define  $\TR^f(y)=\NW^{\cC(f)}(y)$.
Then $\TR^f$ is a black-box $(\epsilon,p)$-PRG with $a+\log L$ advice bits.
\end{lemma}

\begin{proof}
Suppose $T$ $\eps$--breaks the PRG $\TR^f=\NW^{\cC(f)}$. W.l.o.g. we can assume $T$ is deterministic.
Let $\bar{f}=\cC(f) \in \B^\bn$. Given the right advice $adv=adv(f,T)$ to $R$,
$R^T(adv,\cdot)$ is a deterministic function computing $\bar{f}_i$
with average success probability $p$ over $i \in [\bn]$, and using only one query to $T$. The advice to the new reconstruction algorithm $R'$ includes the string $adv$. $R'$ uses the reconstruction algorithm $R^T(adv,\cdot)$ on each $j \in [\bn]$. The resulting string $\wh{y} \in \B^\bn$ has $(\half+\delta)\bn$ agreement with $\bar{f}$. We now use the list decoding algorithm to get a list of up to $L$ codewords in $\cC$ that are $\half+\delta$ close to $\wh{y}$. We know $f$ is the list. By adding $\log(L)$ bits to the advice, we can let the advice tell us which of the codewords in the list is $f$. We have recovered $f$ using $a+\log L$ advice bits and $\bn$ queries.
\end{proof}

Trevisan could tolerate $\bn$ queries. We, however,
in light of Proposition \ref{prop:black-box:extractor-against-quantum-storage},
need to reduce the number of queries. We still want, however, a \emph{worst-case} reconstruction. The idea is to take $\cC$ to be a locally list-decodable code. As our oracle is a probabilistic function what we actually need is a probabilistic oracle, locally list-decodable code. This leads to:

\begin{lemma}(worst-case to average-case reduction for black-box PRG using only few queries)
\label{lem:black-box:average-to-worst-case:few-queries}
Assume $(G^f,R)$ is a black-box $(\epsilon,\half+\delta)$-PRG with
average-case reconstruction using $a$ advice bits. Let $\cC$ be
a $(p=\half+\delta,L,q,\beta)$ probabilistic oracle, local list-decodable binary code.
Define  $\TR^f(y)=\NW^{\cC(f)}(y)$.
Then $\TR^f$ is a black-box $(\epsilon,\beta)$-PRG with $a+\log L$ advice bits and $q$ queries.
\end{lemma}

\begin{proof}
Suppose $T$ $\eps$--breaks the PRG $\TR^f=\NW^{\cC(f)}$.
Let $\bar{f}=\cC(f)$. Given the right advice $adv=adv(f,T)$ to $R$,
$R^T(adv,i)$ computes $\bar{f}_i$ with
average success probability $p=\half+\delta$ over $i \in [\bn]$ and a single query to $T$. The advice to the new reconstruction algorithm $R'$ includes the string $adv$.

Now assume we ask $R'$ for the value of $f_i$, $i \in [n]$, i.e., we wish to compute $R'^T(adv,i)$. We do that as follows. We apply the probabilistic oracle, local list-decoding algorithm of $\cC$, and get $q$ queries $i_1,\ldots,i_q \in [\bn]$ to $\bar{f}=\cC(f)$. We answer the $j$'th query with the probabilistic oracle
$R^T(adv,i_j)$ and we output the decoding result.
By the probabilistic oracle, local list-decoding property, for every $i \in [n]$ the reconstruction oracle $R'^T$, with additionally the right $k \in [L]$, outputs the right answer with probability at least $\beta$.
\end{proof}

Putting it together, we prove Theorem \ref{thm:black-box-PRG:few-queries}

\begin{proof}
Let $\eps>0,  m \le n$.
Let $\NW^{f:[n] \to \B}:\B^{t'} \to \B^m$ be the Nisan-Wigderson PRG with $a=O(m^2)$ advice bits and $t'=O({\log ^2n \over \log m})$. Nisan and Wigderson showed that $\NW^f$ is a black-box $(\eps,\half+\delta)$ PRG with average reconstruction and $\delta={\eps \over 2m}$.

Let $\cC$ be the $(p=\half+\delta,L=poly(\bn),q=poly(\log n,{1 \over \delta}),\beta=1-\delta)$ probabilistic oracle, local list-decodable binary code of Theorem \ref{thm:stv:probabilistic-oracle}. Define $\TR^{f:[n] \to \B}:\B^{t'} \to \B^m$ by $\TR^f(y)=\NW^{\cC(f)}(y)$ with $t'=O({\log^2 {n \over \delta} \over \log m})$. By Lemma \ref{lem:black-box:average-to-worst-case:few-queries} $\TR^f$ is a black-box $(\eps,1-\delta)$ PRG with $a=O(m^2+\log {n \over \eps})$ advice bits and $q$ queries.
\end{proof}

\section{Open problems}

The ideal solution to the problem of classical extractors against quantum storage, is to find a natural, generic transformation from a strong extractor to a strong extractor against quantum storage with about the same parameters. Gavinsky et. al. \cite{GKKRW07} showed this is impossible. Is there a natural class of constructions that does hold against quantum storage? Even if not, a natural objective is to prove that many of the current explicit extractors (and in particular \cite{LRVW03,GUV07,DW08}) are good even against quantum storage.

The parameters given in Theorem \ref{thm:main} can probably be improved. It would be interesting to construct an extractor against quantum storage with logarithmic seed length and arbitrarily small polynomial error, as this may serve as a building block in other constructions.

\section*{Acknowledgements}
I would like to thank Avi Ben-Aroya, Ashwin Nayak, Oded Regev and Pranab Sen for stimulating talks on the subject. I also thank Oded for the example showing that the bound of Theorem \ref{thm:random-access-codes} is tight.

\bibliographystyle{plain}
\bibliography{refs}
\end{document}